\documentclass[10pt,a4paper]{article}
\usepackage{indentfirst}
\usepackage{amsmath,amssymb,amsfonts,amsthm,graphics}
\usepackage{epsfig}
\usepackage{makeidx}
\usepackage{mathrsfs}
\usepackage{color}
\usepackage{cite}
\usepackage{dsfont}
\usepackage{mathrsfs}
\usepackage{algorithmic}
\usepackage{algorithm}
\usepackage{graphicx}
\usepackage{geometry}
\usepackage{epstopdf}
\usepackage{setspace}
\usepackage{caption}
\usepackage{subfig}
\usepackage[section]{placeins}
\usepackage{float}
\newtheorem{theorem}{Theorem}[section]

\newtheorem{proposition}[theorem]{Proposition}

\newtheorem{corollary}[theorem]{Corollary}

\newtheorem{remark}[theorem]{Remark}

\newtheorem{observation}[theorem]{Observation}
\newtheorem{conjecture}{Conjecture}[section]
\newtheorem{problem}[conjecture]{Problem}

\begin{document}
\title{\Large\bf Complexity and Algorithm for the Matching vertex-cutset Problem }

\author{Hengzhe Li$^{a}$, Qiong Wang$^{a}$, Jianbing Liu$^{b}$, Yanhong Gao$^{a}$\\
\small $^{a}$College of Mathematics and Information Science,\\
\small Henan Normal University, Xinxiang 453007, P.R. China\\
\small $^b$ Department of Mathematics, \\
\small University of Hartford, West Hartford 06117, USA\\
\small Email: lihengzhe@htu.edu.cn, wang1472581113@163.com,\\ \small jianliu@hartford.edu, gaoyanhong@htu.edu.cn}
\date{}
\maketitle
\begin{abstract}
In 1985, Chv\'{a}tal introduced the concept of star cutsets as a means to investigate the properties of perfect graphs, which inspired many researchers to study cutsets with some specific structures, for example, star cutsets, clique cutsets, stable cutsets. In recent years, approximation algorithms have developed rapidly, the computational complexity associated with determining the minimum vertex cut possessing a particular structural property have attracted considerable academic attention.

In this paper, we demonstrate that determining whether there is a matching vertex-cutset in $H$ with size at most $k$, is $\mathbf{NP}$-complete, where $k$ is a given positive integer and $H$ is a connected graph. Furthermore, we demonstrate that for a connected graph $H$, there exists a $2$-approximation algorithm in $O(nm^2)$ for us to find a minimum matching vertex-cutset. Finally, we show that every plane graph $H$ satisfying $H\not\in\{K_2, K_4\}$ contains a matching vertex-cutset with size at most three, and this bound is tight.
\end{abstract}

\section{Introduction}
In the whole paper, every graph is finite, undirected, and simple. For any undefined notation and terminology, we refer to the books \cite{Bondy} and  \cite{Du}. For a graph $H=(V(H), E(H))$, its {\it minimum degree}, {\it maximum degree} and {\it connectivity} are denoted by $\delta(H)$, $\Delta(H)$, and $\kappa(H)$, respectively. A {\it path} is often regarded as an $(u, v)$-path if its ends are $u$ and $v$. For some vertex $x\in V(H)$, $N_{H}(x)$=$\{y\in V(H):xy\in E(H)\}$, $N_{H}[x]=N_H(x)\cup\{x\}$.
The {\it degree} of $x$ in $H$ is the size of $N_{H}(x)$.

For $X\subseteq V(H)$, the {\it induced subgraph} $H[X]$ by $X$ is the subgraph of $H$ with vertex set $X$, and edge set $\{uv\in E(H): u,v\in X\}$. For any $Y, Z\subseteq V(H)$, $E_H[Y,Z]=\{yz\in E(H):y\in Y\ and\ z\in Z\}$.

A {\it matching} in a graph is a set of non-adjacent edges. If $M$ is a matching, then every vertex adjacent to an edge of $M$ is be {\it covered} by $M$. We represent the set of these vertices by $V(M)$. For a matching $M$ of a graph $H$, if $V(M) = V(H)$, then we call it a {\it perfect matching}; if $|M|$ is maximized, then we call it a {\it maximum matching}. A {\it maximal matching} is a larger matching that can not be extended any more.

Let $S$ be a vertex set of a connected graph $H$, if $H$ becomes disconnected or trivial after deleting $S$, we regard $S$ as a {\it vertex cut}; if there exists $u\in S$ satisfying $S\subseteq N_H[u]$, then we call $S$ a {\it star cutset}; if $H[S]$ is a complete graph, then we regard $S$ as a {\it clique cutset}. {\it A stable set} of $H$ is a set of vertices where any two members are not adjacent, we also call stable sets {\it independent sets}. Similarly, if $H[S]$ is a stable set, then we call $S$ a {\it stable cutset}.

Many researchers have also paid attention to vertex cut-sets with some specific structures. For example, in 1985, Chv\'{a}tal\cite{Chvatal} introduced star cutsets to study perfect graphs. Whitesides \cite{Whitesides} offered an $O(n^3)$-time algorithm to find a clique cutset. Brandstadt\cite{Brandstadt} showed that whether a connected graph possesses a stable cutset is an $\mathbf{NP}$-complete problem.

Edge cut-sets with some specific structures  have also drawn the interest of researchers. A {\it matching-cut} $M$ of a connected graph $H$ is a matching $M$ of $H$ and satisfies that $H-M$ becomes disconnected. Patrignani and Pizzonia\cite{Patrignani} demonstrated that it is $\mathbf{NP}$-complete to determine whether a connected graph $H$ contains a matching cut, and the application of matching cuts in graph drawings were discussed. The matching cut of the network application environment diagram were studied by Farley and Proskurowski\cite{Farley}.

Lin et al. \cite{Lin} put forward the notions of structure connectivity in 2016. Let $F$ be a connected subgraph of a connected graph $H$, we represent the {\it $F$-structure connectivity} of the graph $H$ by $\kappa(H; F)$, is the minimum size of a set of disjoint subgraphs $\mathcal{F}$=$\{f_1$,$f_2$,$\cdots$,$f_{m}\}$ in $H$, satisfying each $f_{i}\in\mathcal{F}$ is isomorphic to $F$ and $H-\cup_{f_{i}\in\mathcal{F}}V(f_{i})$ is disconnected or trivial. If $F\cong P_2$, then we regard $P_2$-structure cutset as {\it matching vertex-cutset} and regard $P_2$-structure-connectivity as {\it matching connectivity}, and $\kappa(H, P_2)$ is often written by $\kappa_M(H)$. The structure connectivity of several prominent networks, for example, \cite{Chelvam, Dilixiati, Feng} et al., have been established recently.

Matching vertex-cutsets and matching-cuts share similarities as they are both related to matching, but these two concepts also have significant differences. For a matching $M$ in a graph $H$, $M$ being matching-cut means that $H-M$ is disconnected, while $H[M]$ being matching vertex-cutset means that $H-V(M)$ becomes disconnected or trivial.

Let $M$ be a matching such that $H[M]$ is a matching vertex-cutset. For convenience in discussion, we are slightly misusing notation by referring to
$M$ as the matching vertex-cutset, even though $M$ is actually a set of edges.

Inspired by the study of vertex cuts with some structures, we study the matching vertex-cutset in our paper. The content in this thesis is as follows. In Section two, we summarize some notations as well as known theorems.
In Section three, we show that after giving a graph $H$ and a positive integer $k$, determining whether there is a set $M$ of at most $k$ independent edges satisfying that $H-V(M)$ becomes disconnected or trivial is $\mathbf{NP}$-complete.  What's more, we demonstrate that there exists a $2$-approximation algorithm in $O(nm^2)$ to find a minimum matching vertex-cutset in a graph $H$. In Section four, we demonstrate that each plane graph $H$ satisfying $H\not\in\{K_2,K_4\}$ possesses a matching vertex-cutset with size at most three and the bound is sharp; the lower bound of the matching connectivity of maximal planar graphs is $2$ and the bound is sharp.

\section{Preliminaries}
We sum up several basics used in our discussion.

In a graph $H=(V, E)$, for a set of edges $M$, if all edges in $E-M$ are associated with some member of $M$, then $M$ is an {\it edge dominating set}. If any two edges in edge dominating set are non-adjacent, then we call it {\it independent edge dominating set}.

For a graph $H$, if $V(H)$ can be divided into two subsets $U$ and $V$ satisfying that one end of each edge is in $U$ and the other end is in $V$, then we call $H$ a {\it bipartite graph}, $(U, V)$ are bipartitions of $H$. Moreover, if for all vertices pairs $(u, v)$, where $u\in U, v\in V$, $uv\in E(H)$, then we call $H$ a {\it complete bipartite graph}. The following theorem presents an equivalent condition about the property of a bipartite graph.

\begin{theorem}[\upshape Hall~\cite{Hall}]\label{Hall}
There exists a matching covering all vertices of $U$ in a bipartite graph $H=H[U, V]$ is equivalent to that $|N_{H}(S)|\geq|S|$ for all $S\subseteq U$.
\end{theorem}

Let $M$ be a matching in a graph $H$, $P$ be a path of $H$, if each edge in $P$ occurs in $M$ and $E\setminus M$ in turn, then we call $P$ an {\it $M$-alternating path}. If any ends are not covered by $M$, then $P$ is called an {\it $M$-augmenting path}. Berge provided an equivalent condition about a maximum matching in \cite{Berge}.

\begin{theorem}[\upshape Berge~\cite{Berge}]\label{Berge}
A matching $M$ in a graph $H$ is maximum when and only when $H$ possesses no $M$-augmenting path.
\end{theorem}

\section{Complexity and algorithm}
\subsection{$\mathbf{NP}$-complete}
The following $\mathbf{NP}$-complete problem will be used in our proof.

{\bf The edge dominating set problem:} After giving a positive integer $k$ and a graph $H$, determine whether there exists an edge dominating set with at most $k$ edges in $H$.

\begin{theorem}[\upshape Yannakakis and Gavril~\cite{Yannakakis}]\label{edgedominatingproblem}
Given a positive integer $k$ and a bipartite graph $H$ with $\Delta(H)=3$, determine whether there exists an edge dominating set with at most $k$ edges in $H$
is $\mathbf{NP}$-complete.
\end{theorem}

In \cite{Yannakakis}, we know that when a graph possesses a minimum edge dominating set with size $k$, $k$ is a positive integer, the size of its minimum independent edge dominating set is also $k$. Yannakakis and Gavril showed that after giving a minimum edge dominating set, they are able to use a polynomial time algorithm to construct a minimum independent edge dominating set, and in their paper, they presented that algorithm. Thus the above problem is amount to the independent edge dominating set problem.

{\bf The independent edge dominating set problem:} After giving a positive integer $k$ and a graph $H$, determine whether there exists an independent edge dominating set with at most $k$ edges in $H$.

By the equivalency of above two problems, the following corollary holds.

\begin{corollary}[\upshape Yannakakis and Gavril~\cite{Yannakakis}]\label{independentedgedominatingproblem}
Given a positive integer $k$ and a bipartite graph $H$ with $\Delta(H)=3$, determine whether there is an independent edge dominating set with at most $k$ edges in $H$
is $\mathbf{NP}$-complete.
\end{corollary}

We present the matching vertex-cutset problem as follows.

{\bf The matching vertex-cutset problem:} After giving a positive integer $k$ and a graph $H$, determine if there exists a matching vertex-cutset in $H$ with size at most $k$.

For matching vertex-cutsets, the following observation holds.

\begin{observation}\label{Lowerbound}
If a graph $H$ possesses a matching vertex-cutset, then $\kappa(H)\le 2\kappa_M(H)$.
\end{observation}

\begin{theorem}\label{theorem1}
The matching vertex-cutset problem for a graph which is neither $K_{2n}$ nor $K_{n, n}$ is $\mathbf{NP}$-complete.
\begin{proof}
Since one can examine polynomial time if a matching is a matching vertex-cutset, the matching vertex-cutset problem for a graph which is neither $K_{2n}$ nor $K_{n,n}$ is in $\mathbf{NP}$. It remains to show that the matching vertex-cutset problem is $\mathbf{NP}$-hard, and we shall demonstrate it by reducing the edge dominating set problem to the matching vertex-cutset problem.

Given a positive integer $k$ and a bipartite graph $H$ with bipartition $X$ and $Y$ satisfying maximum $3$ and $k\le min\{|X|,|Y|\}$, our idea is to construct a new graph $H'$ based on the graph $H$, and show that there is a matching vertex-cutset with size at most $k$ in $H'$ is equivalent to that there is an independent edge dominating set with size at most $k$ in $H$. The graph $H'$ we constructed is shown below. On the basis of the graph $H$, we add two stable sets $X'$ and $Y'$ satisfying $|X'|=|X|$ and $|Y'|=|Y|$, add all edges between $X$ and $X'$, between $Y$ and $Y'$, satisfying for any $x\in X, x'\in X', xx'\in E(H')$, for $y\in Y$, $y'\in Y'$, $yy'\in E(H')$. See Figure~1 for details.

\begin{center}
\scalebox{0.5}[0.5]{\includegraphics{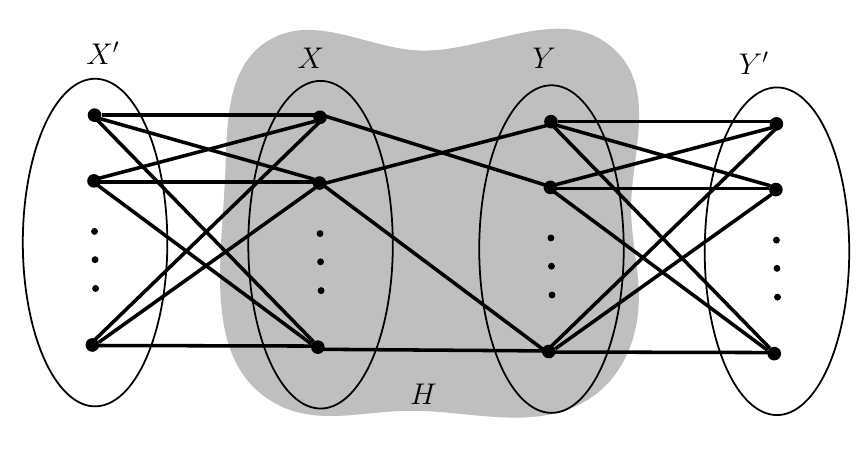}}

\small
Figure 1. The graphs $H$ and $H'$.
\end{center}

On one hand, if there exists an independent edge dominating set $M$ with $|M|\leq k$ in $H$, then $H-V(M)$ is disconnected or trivial, that means, $M$ is a matching vertex-cutset of $H'$ by the definition of matching vertex-cutset. So there exists a matching vertex-cutset with size at most $k$ in $H'$.

On the other hand, suppose that the minimum size of a matching cutset $M'$ in $H'$ is at most $k$, but the minimum size of an independent edge dominating set of $H$ is larger than $k$. Then $min\{|X|,|Y|\}\geq k+1$, so $|X'|, |Y'|\geq k+1$ since $|X'|=|X|$ and $|Y'|=|Y|$.
Hence $X\setminus V(M')\neq\emptyset$, $X'\setminus V(M')\neq\emptyset$, $Y\setminus V(M')\neq\emptyset$, and $Y'\setminus V(M')\neq\emptyset$.
Because $H[X\cup X']$ and $H[Y\cup Y']$ are complete bipartite graphs, the graph $H[X\cup X']\setminus V(M)$ and $H[Y\cup Y']\setminus V(M)$ are connected. Moreover, the graph $H-V(M')$ has at least one edge, since $M'$ is a matching in $H$ with size at most $k$, and the minimum size of an independent edge dominating set of $H$ is larger than $k$. Thus $H'-V(M')$ is connected, a contradiction.

The proof is end.
\end{proof}
\end{theorem}

\subsection{The algorithm}
In this subsection, we provide an approximation algorithm for the following problem.

{\bf The minimum matching vertex-cutset problem:} In a graph $H$, where $H$ is neither $K_{2n}$ nor $K_{n,n}$, finding its minimum matching vertex-cutset.

For the minimum matching vertex-cutset problem, we have a $2$-approximation algorithm.

\begin{theorem}\label{Algorithmvalid}
For the minimum matching vertex-cutset problem, there exists a $2$-approximation algorithm in $O(nm^{2})$.
\end{theorem}
\begin{proof}
Let $H$ be a simple connected graph which is neither $K_{2n}$ nor $K_{n,n}$.
By running Max-Flow Min-Cut Algorithm \cite{Ford}, we can find a minimum vertex cut $S$ in $O(nm^{2})$. For such vertex cut $S$, we would like to construct a matching vertex-cutset $M$ with $|M|\le |S|$ in $O(nm^{2})$. By Observation \ref{Lowerbound}, our algorithm has the approximation ratio $2$.

Let $U$ be a connected component of $H-S$, and let $V=H-S-U$.
By running Hungary Algorithm \cite{Edmonds}, we are able to find a maximum matching $M_1$ of $H[S]$ in $O(nm)$. If $V(M_1)=S$, then $M_1$ is a matching vertex-cutset of $H$, and we are done; otherwise, let $S_1=S\cap V(M_1)$, $S_2=S\setminus S_1$. Since $S_1$ is a maximum matching, the set $S_2$ is an independent set of $H$.

If $|S_2|=1$, then we assume that $x\in S_2$. Since $S$ is the minimum vertex cut of $H$, we have that $N_H(x)\cap U\neq\emptyset$ and $N_H(x)\cap V\neq\emptyset$. If $|U|\ge 2$, then pick $x'\in N_H(x)\cap U$, and $M_1\cup \{xx'\}$ is a matching vertex-cutset. Otherwise, $x'\in N_H(x)\cap V$, and $M_1\cup \{xx'\}$ is a matching vertex-cutset.

If $|S_2|\ge 2$, we give a thought to the maximal bipartite subgraph $H'$ of $H$ with bipartition $(S_2, U\cup V)$.
For a vertex $x\in S_2$, $|N_H(x)\cap (U\cup V)|\geq|S_2|$ because
$|S_1|+|N_H(x)\cap (U\cup V)|\geq d_H(x)\geq\delta(H)\geq\kappa(H)=|S|=|S_1|+|S_2|$. Hence, $H'$ has a matching that covers $S_2$ by Theorem \ref{Hall}. As $S$ is a minimum vertex cut, one can pick $x,y\in S_2$, $x'\in U$, and $y'\in V$ satisfying $xx'\in E(H)$ and $yy'\in E(H)$.
By running Hopcroft-Karp Algorithm \cite{Hopcroft} in $O(nm)$, one can find such a matching $M_2$ from $\{xx',yy'\}$ that covers $S_2$ in $O(m\sqrt{n})$. Let $A=S_2\cap V(M_2)$, $B=S_2\setminus A$, $U_1=U\cap V(M_2)$, $U_2=U\setminus U_1$, $V_1=V\cap V(M_2)$, $V_2=V\setminus V_1$, $M_A=M_2\cap E_H[U,S]$ and $M_B=M_2\cap E_H[V,S]$.

Since $M_2$ is obtained from $\{xx',yy'\}$ by running Hopcroft-Karp Algorithm \cite{Hopcroft}, the sets
$M_A$ and $M_B$ are not empty.

{\it Case $1$.} $U_2\neq\emptyset$ and $V_2\neq\emptyset$.

Then $M_1\cup M_2$ is a matching vertex-cutset of $H$, and we output $M_1\cup M_2$.

{\it Case $2$.} $U_2=\emptyset$ and $V_2\neq\emptyset$.

   {\it Subcase $2.1$.} $|V_2|=1$.

   Then $M_1\cup M_2$ is a matching vertex-cutset of $H$, we output $M_1\cup M_2$. Now, $H-V(M_1\cup M_2)$ is travial.

   {\it Subcase $2.2$. $|V_2|\geq2$.}

   If $E_H[A,V_2]=\emptyset$, then $M_1\cup M_B$ is a matching vertex-cutset of $H$, we output $M_1\cup M_B$. Otherwise, that is, $E_H[A,V_2]\neq\emptyset$, then choose an edge $xy\in E_H[A,V_2]$ with $x\in A$, $y\in V_2$, and $(M_1\cup M_2\cup\{xy\})\setminus\{xx'\}$ is a matching vertex-cutset of $H$,
   where $xx'\in M_A$.

{\it Case $3$.} $U_2\neq\emptyset$ and $V_2=\emptyset$.

 One can solve this case similar to Case~$2$ by just switching  $V_2$ and $U_2$.

{\it Case $4$.} $U_2=\emptyset$ and $V_2=\emptyset$.

   {\it Subcase $4.1$.} $U\cup V$ is a stable set in $H$.

   When both $U_2=\emptyset$ and $V_2=\emptyset$. Under the circumstances, we have $U=U_1$ and $V=V_1$. For any $x\in S_2$ and $y\in S_1\cup U\cup V$, as $d_H(x)\geq\delta(H)\geq\kappa(H)=|S_1|+|S_2|=|S_1|+|U|+|V|$ and as $S_2$ is an independent set, we conclude that for any $x\in S_2$, $S_1\cup U\cup V\subseteq N_H(x)$. Therefore, $H[U\cup V,S_2]$ is a complete bipartite graph. What's more, we know that $S_1=\emptyset$. If not, then we pick $x,y\in S_1$ satisfying $xy\in M_1$, and pick $x',y'\in S_2$. By $H[U\cup V,S_2]$ is a complete bipartite graph, we have that $xx', yy'\in E(H)$, that is, $x'xyy'$ is an augmenting path in $H[S]$ in connection with $M_1$, a contradiction to Theorem \ref{Berge} and the hypothesis that $M_1$ is a maximum matching in $H[S]$. So $H$ is a complete bipartite graph with equal bipartition $(U\cup V,S_2)$, that is, $H\cong K_{n,n}$, when $n:=|U\cup V|=|S|$. There is a contradiction with our hypothesis that $H$ is not $K_{n,n}$, we output $\emptyset$.

  {\it Subcase $4.2$.} $U\cup V$ is not a stable set in $H$.

   Without sacrificing generality, since there exists no edge between $U$ and $V$, we pick $x,y\in U$ satisfying $xy\in E(H)$, let $xx', yy'\in M_A$. Then $(M_1\cup M_2\cup\{xy\})\setminus\{xx',yy'\}$ is a matching vertex-cutset of $H$, we output $(M_1\cup M_2\cup\{xy\})\setminus\{xx',yy'\}$.

   Above all, our algorithm is correct and run in $O(nm^{2})$.
\end{proof}

\begin{remark}\label{Existence}
On account of the proof of Theorem~\ref{Algorithmvalid}, one can see that each connected simple graph which is neither $K_{2n}$ nor $K_{n,n}$ has a matching vertex-cutset.
\end{remark}

\section{The matching connectivity of plane graphs}
In this section, we study the matching connectivity of plane graphs.

For a plane graph $H$, let $F(H)$ be the set of all faces in $H$.
The {\it degree} $d(f)$ of a face $f$ is the number of edges in its boundary, cut edges being counted twice.

The following two results will be used in our proof.

\begin{theorem}[\upshape Euler~\cite{Euler}]\label{Euler 1}
In a connected plane graph $H$, we have that $|V(H)|-|E(H)|+|F(H)|=2$.
 \end{theorem}

\begin{theorem}[\upshape Euler~\cite{Euler}]\label{Euler 2}
In a plane graph $H$, we have that $\Sigma_{f\in F}d(f)=2|E(H)|$.
\end{theorem}

On account of the proof Theorem \ref{Algorithmvalid}, we get a proposition.

\begin{proposition}\label{Upperbound}
If there exists a matching vertex-cutset in a graph $H$, then $\kappa_M(H)\le \delta(H)$.
\end{proposition}

The method of discharging plays a significant role in proving the Four-Color theorem. Initially, each vertex $v$ is arranged with a weight, such as $d(v) - 4$, that is, {\it charge} of the vertex; every face $f$ is arranged with a weight, such as $d(f) - 4$, that is, {\it charge} of the face. The charges are arranged satisfying that the number they add up to is a positive (or negative) number. We are trying to discharge $ H $, making the charge at every vertex and each face non-positive (or non-negative) by redistributing the charges. Since there is no charge increase or decrease, but just redistribution, this process can lead to a contradiction.
Now, we ready to bound the matching connectivity of a plane graph.

\begin{theorem}\label{planargraphupper}
If $H$ is a connected plane graph satisfying $H\not\in\{K_2,K_4\}$,
then $\kappa_{M}(H)\leq 3$.
\end{theorem}
\begin{proof}
By Remark \ref{Existence}, each simple connected plane graph which satisfies $H\not\in\{K_2,K_4\}$ has a matching vertex-cutset. Suppose, by way of contradiction, $H$ is a simple connected plane graph with $\kappa_{M}(H)\geq4$. We have $|V(H)|\ge 9$.

Our proof uses the discharging technique. We multiply both sides of Euler's formula by $4$, we can have that in a plane graph $H$, $(2|E(H)|-4|V(H)|)+(2|E(H)|-4|F(H)|)=-8$. By Theorem \ref{Euler 2}, this formula amounts to $\Sigma_{v\in V}(d(v)-4)+\Sigma_{v\in V}(d(f)-4)=-8$.

Firstly, for each vertex $v$ in $H$, we arrange $d(v)-4$ units of charge to it, and for every face $f$ in $H$, we arrange $d(f)-4$ units of charge to it. By this time, the sum of all the charges is $-8$, that is, negative. We stipulate some rules to rearrange the charges so that after the rearrangement, if the charge of each vertex and face is non-negative, but there is no charge increase or decrease in the graph, then there is a contradiction.

The rules for charge rearrangement is that every vertex $v\in V$ assigns $\frac{1}{3}$ unit of charge to every triangular face $f$ associated with $v$.

We now observe that every face has a nonnegative charge. We consider the degree of every face. If $d(f)\geq4$, then $d(f)-4\geq0$, face $f$ has a non-negative charge. If $f$ is a triangle, that is, $d(f)=3$, then $d(f)-4=-1$, as it is incident with exactly $3$ vertices, so it gains $\frac{1}{3}$ unit from each of them, and it has $-1+1=0$ unit of charge, so it is nonnegative.

Next, we consider the charge of every vertex after discharging. By assumption and Proposition \ref{Upperbound}, $\delta(H)\ge \kappa_{M}(H)\ge 4$. For some vertex $x\in V(H)$, let $t(x)$ denote the  number of triangular faces incident to $x$. Pick any $v\in V(H)$, we show that each vertex has a nonnegative charge after discharging.

If $d_H(v)=4$, then we assume that $N_G(v)=\{v_1, v_2, v_3, v_4\}$. We say that $v_iv_j\not\in E(H)$, where $1\leq i<j\leq 4$, that means, $v$ can not be incident with some triangular face; otherwise, without sacrificing generality,  assume $v_1v_2\in E(H)$. As $\delta(H)\geq4$, one can pick $v'_3\in N_H(v_3)\setminus\{v\}, v'_4\in N_H(v_4)\setminus\{v\}$.
Then $\{v_1v_2,v_3v'_3,v_4v'_4\}$ is a matching vertex-cutset of $H$ which is inconsistent with our hypothesis, $\kappa_{M}(H)\geq4$. So, after discharging, we have $d(v)-4=4-4=0$ .

If $d_H(v)=5$, then we assume that $N_H(v)=\{v_1, v_2, v_3, v_4, v_5\}$. We say that $v$ is incident with at most $2$ triangular faces, otherwise, there are distinct $i,j,k,\ell\in\{1,\ldots,5\}$ satisfying $v_iv_j, v_kv_{\ell}\in E(H)$. For $v_r\in N_H(v)\setminus\{v_i, v_j, v_k, v_{\ell}\}$, Since $\delta(H)\geq4$, one can pick $v'_r\in N_H(v_r)\setminus\{v\}$.
Then $\{v_iv_j, v_kv_{\ell}, v_rv'_r\}$ is a matching vertex-cutset of $H$ which is inconsistent with our hypothesis, $\kappa_{M}(H)\geq4$. After discharging, $v$ should give $\frac{1}{3}$ unit charge to every triangular face incident to it. Since $t(v)\leq 2$, we have $d(v)-4-\frac{1}{3}t(v) \geq\frac{1}{3}>0$.

If $d_H(v)\ge 6$, then $v$ is incident with at most $d_H(v)$ triangular faces, that is, $t(v)\leq d_H(v)$. After discharging, $v$ should give $\frac{1}{3}$ unit charge to each triangular face incident to it. As $t(v)\leq d_H(v)$, we have $d(v)-4-\frac{1}{3}t(v)\geq \frac{2}{3}d(v)-4\geq 0$.

Above all, every vertex and every face has a non-negative charge after discharging, which contradicts $\Sigma_{v\in V}(d(v)-4)+\Sigma_{v\in V}(d(f)-4)=-8$.

We complete our proof.
\end{proof}

\begin{remark}
For the icosahedron $G$ in Figure~2, on one hand, since $G$ is $5$-connected graph, $\kappa_M(G)\ge 3$. On the other hand, the matching $\{e_1,e_2,e_3\}$ is a matching vertex-cutset. Thus, $\kappa_M(G)=3$, and the upper bound is sharp in Theorem~\ref{planargraphupper}
\end{remark}

\begin{center}
\scalebox{0.5}[0.5]{\includegraphics{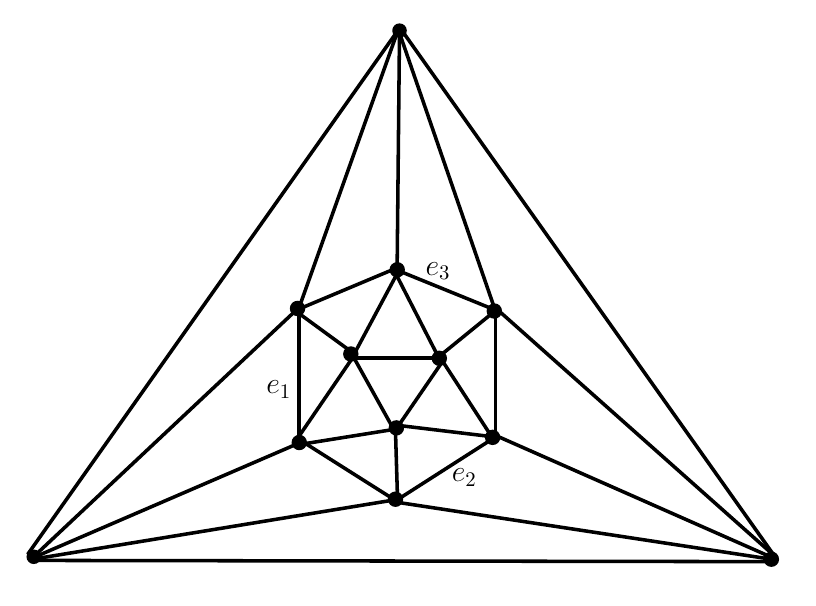}}

\small
Figure 2. The icosahedron $H$ and a matching vertex-cutset with size three.
\end{center}

A plane graph $H$ is a {\it maximal planar graph} if $H+uv$ is not a planar graph for any $u, v\in V(H), uv\notin E(H)$.

\begin{theorem}\label{planargraphlower}
If $H$ is a maximal planar graph, then $\kappa_M(H)\geq2$.
\end{theorem}

\begin{proof}
Since $H$ is a maximal planar graph, we know that $H$ is connected, which implies that $\kappa_M(H)\geq1$. We demonstrate that $\kappa_M(H)\ge 2$.
Otherwise, suppose that $\{uv\}$ is a matching vertex-cutset.
Let $\tilde{H}$ be a planar embedding of $H$, and let $\tilde{H}_1, \tilde{H}_2, \cdots, \tilde{H}_m$ be the components of $\tilde{H}-\{u, v\}$, $m\geq2$.
For $i=1,2$, taking vertices $v_i\in V(H_i)$ such that $v_i$ lies bound of the outer face of $\tilde{H}$. Then $H+v_1v_2$ is a planar graph, a contradiction
\end{proof}

\begin{remark}
$K_5^-$ is a maximal planar graph, and $\kappa_M(K_5^-)=2$, where $K_5^-$ is the graph obtained from $K_5$ by deleting an edge. Thus, the lower bound is sharp in Theorem~\ref{planargraphlower}.
\end{remark}
\section{Conclusion}

In our research, we have demonstrated that the matching vertex-cutset problem for a graph that is neither $K_{2n}$ nor $K_{n,n}$ is $\mathbf{NP}$-complete. Moreover, we demonstrate that there is a $2$-approximation algorithm in $O(nm^2)$ to find a minimum matching vertex-cutset in a graph $H$. We summarize the paper by proposing two problems.

\begin{problem} Desiging a polynomial-time algorithm with an approximation ratio of less than $2$ for the minimum matching vertex-cutset problem.
\end{problem}

\begin{problem} Designing a more efficient $2$-approximation algorithm for the minimum matching vertex-cutset problem.
\end{problem}

\section{Acknowledgements}
This work was supported by the National Natural Science Foundation of China (No. 12101203) and China Postdoctoral Science Foundation (No. 2022M711075).

\end{document}